\documentclass[12pt]{article}
\usepackage[top=2.5cm,left=2.5cm,right=2.5cm,bottom=2.5cm]{geometry}

\usepackage[dvips]{epsfig}    % or this line, depending on which
\usepackage{t1enc}
\usepackage[english]{babel}

\usepackage{amsmath}
\usepackage{amssymb}
\usepackage{shortvrb,psfrag,graphicx}

\newtheorem{theorem}{Theorem}
\newtheorem{lemma}{Lemma}
\newtheorem{proposition}{Proposition}
\newcommand{\qfd}{\hfill{\rule{1.8mm}{1.8mm}}}

\newenvironment{proof}[1][\bf{Proof.}]{\begin{trivlist}
\item[\hskip \labelsep {\bfseries #1}]}{\end{trivlist}}

\newcommand{\Xomit}[1]{}

\newcommand{\matA}{{\mathbf{A}}}
\newcommand{\matB}{{\mathbf{B}}}
\newcommand{\matP}{{\mathbf{P}}}
\newcommand{\matQ}{{\mathbf{Q}}}
\newcommand{\matR}{{\mathbf{R}}}
\newcommand{\matS}{{\mathbf{S}}}
\newcommand{\matT}{{\mathbf{T}}}

\newcommand{\matnull}{{\mathbf{0}}}

\newsavebox{\traitbox}

\begin{document}


\title{Bargaining Dynamics in Exchange Networks\footnote{A short version without proofs was presented as an invited paper at Allerton 2010.}} % Title, preferably not more 
%                                                % than 10 words.

\author{Moez Draief\footnote{Department of Electrical and Engineering, Imperial College, London, United Kingdom, m.draief@imperial.ac.uk}  and    Milan Vojnovi\'c\footnote{Microsoft Research, Cambridge, United Kingdom, milanv@microsoft.com} }  % e-mail address 

   \maketitle
\begin{abstract}                          % Abstract of not more than 200 words.
We consider a dynamical system for computing Nash bargaining solutions on graphs and focus on its rate of convergence. More precisely, we analyze the edge-balanced dynamical system by Azar et al and fully specify its convergence for an important class of elementary graph structures that arise in Kleinberg and Tardos' procedure for computing a Nash bargaining solution on general graphs. We show that all these dynamical systems are either linear or eventually become linear and that their convergence times are quadratic in the number of matched edges. 
\end{abstract}

\section{Introduction}
\label{sec:intro}

Bargaining and, in particular, the concept of Nash bargaining on general graphs has been the focus of much recent research in economics, sociology and computer science~\cite{ABCDP09,CKK09,KBBCM10b,KBBCM10,KT08}. In a bargaining system, players aim at making pairwise agreements to share a fixed wealth specific to each pair of players. Bargaining solutions provide predictions on how the wealth will be shared and how this sharing would depend on players' positions in a network describing some notion of relationships among players.  

The concept of Nash bargaining solution was introduced by Nash~\cite{N50} for two players, each having an exogenous, alternative profit at its disposal were they to disagree on how to share the wealth. Recent research has focused on the concept of Nash bargaining with multiple players where each player has alternative profits determined by trading opportunities with neighbors in a graph. In the computer science literature, Kleinberg and Tardos~\cite{KT08} were the first to establish various properties of Nash bargaining outcomes on general graphs. They also propose a polynomial-time algorithm for computing them, provided one exists. Follow up work aimed at introducing some local dynamics that are natural (so they, hopefully, have some connections with reality) and studied their convergence properties. In particular, Azar et al \cite{ABCDP09} considered the so called \emph{edge-balanced dynamics} and established various properties about fixed points and convergence but left open the characterization of the convergence rate. In a tandem of papers \cite{KBBCM10,KBBCM10b}, Kanoria et al considered an alternative, \emph{natural dynamics}, and established polynomial convergence time bounds under some generic assumptions. An open research question has been to gain a better understanding of convergence properties and obtain tight bounds on the convergence time for these types of systems. 

In this paper, we consider the edge-balanced dynamics of Azar et al~\cite{ABCDP09} over elementary graphs that arise in the decomposition procedure of Kleinberg and Tardos~\cite{KT08} which include a path, a cycle, a blossom and a bicycle (see Figures~\ref{fig:path}, \ref{fig:cycle}, \ref{fig:blossom} and \ref{fig:bi-cycle} for examples). It turns out that, for all these network structures, the dynamics is either linear or eventually becomes linear. Specifically, we show that the dynamics is \emph{linear} for a path and a cycle and is \emph{eventually linear} for a blossom and a bicycle (and characterize the time when this takes place). This allows us to fully characterize the rate of convergence by deploying well known spectral methods for linear systems. As a result, for all these elementary structures, we find that the convergence time is \emph{quadratic} in the number of matched edges.

\subsection{Outline of the Paper} In Section~\ref{sec:sys} we introduce system assumptions and overview relevant concepts, including the concept of Nash bargaining outcomes, local dynamics, and the KT procedure. Section~\ref{sec:KTsub} provides the characterization of the edge-balanced dynamics and convergence times for each of the elementary graphs of the KT decomposition. Finally, we discuss related work and conclude in Sections~\ref{sec:related} and~\ref{sec:conc} respectively. 

\section{System and Assumptions}
\label{sec:sys}

\subsection{Nash Bargaining Outcomes on Graphs}

We consider a graph $G = (V,E)$ where $V$ is the set of nodes and $E$ is the set of edges. Each node corresponds to a distinct player that participates in the trading game defined as follows. Each edge $(i,j)\in E$ is associated with a weight $w_{i,j}\geq 0$ representing the amount that can be shared between players $i$ and $j$ should these two players decide to trade with each other. The trading game is one-exchange meaning that each player attempts to make a pairwise agreement with at most one other player, which corresponds to a matching $M \subset E$ in the graph where $(i,j)\in M$ if and only if players $i$ and $j$ reached an agreement. We denote with $x_i$ the profit of player $i$ where $x_i \geq 0$ and let $\vec{x} = (x_i)_{i\in V}$ denote the vector of players' profits.

A balanced outcome or a Nash bargaining solution is a pair $(M,\vec{x})$ where $M$ is a matching in $G$ and $\vec{x}$ is a vector of players' profits. Such an outcome satisfies the following properties:
\begin{itemize}
\item {\bf Stability}: for every edge $(i,j)\in E$,
$$
x_i + x_j \geq w_{i,j}. 
$$
\item {\bf Balance}: for every $(i,j)\in M$, it holds that
$$
x_i - \max_{k\in V_i\setminus \{j\}}(w_{i,k} - x_k)_+
= x_j - \max_{k\in V_j\setminus \{i\}}(w_{j,k} - x_k)_+
$$ 
where, hereinafter, $V_i$ denotes the set of neighbors of a node $i$ and $(\cdot)_+:=\max(0,\cdot)$. 
\end{itemize}

The stability property means that there exists no player that can improve her profit by unilaterally deciding to trade with an alternative trading partner. The balance property originates from the Nash bargaining problem~\cite{N50} where two players $1$ and $2$ aim at a pairwise agreement to share a profit $w$ having outside profit options $r_1$ and $r_2$ in case of disagreement. The Nash bargaining solution is then for players $1$ and $2$ to share the surplus $w-r_1-r_2$ equally, if positive, i.e. receive profits $p_1 = r_1 + \frac{1}{2}(w - r_1 - r_2)_+$ and $p_2 = r_2 + \frac{1}{2}(w - r_1 - r_2)_+$, respectively. This allocation is balanced in the sense that $p_1 - r_1 = p_2 - r_2$, which is exactly the above asserted balance property where the outside profit options are determined by the values that players may extract through trading agreements with their neighbors.

\subsection{KT Procedure}
\label{sec:positivegap}

Nash's bargaining solutions on graphs are intimately related to maximum-weight matchings. In~\cite{KT08} it was found that the matching $M$ of a stable outcome $\vec{x}$ is a maximum-weight matching. Furthermore, whenever a stable outcome exists, a balanced outcome exists as well~\cite{KT08}. The outcome vector $\vec{x}$ can be seen as a feasible solution of a dual to the fractional relaxation of a maximum-weight matching (primal): 
$$
\begin{array}{rl}
\hbox{maximize} & \sum_{(i,j)\in E} w_{i,j} x_{i,j}\\
\hbox{over} & x_{i,j} \geq 0,\ (i,j)\in E\\
\hbox{subject to} & \sum_{j: (i,j)\in E} x_{i,j} \leq 1
\end{array}
$$
where the dual problem is the following linear problem with two variables per inequality:
$$
\begin{array}{rl}
\hbox{minimize} & \sum_{i\in V} x_i\\
\hbox{over} & x_i\geq 0,\ i\in V\\
\hbox{subject to} & x_i+x_j \geq w_{i,j},\ (i,j)\in E.
\end{array}
$$

In \cite{KT08}, it was established that a balanced outcome $(M,\vec{x})$ can be found in polynomial time by first finding a maximum-weight matching $M$ and then solving the above dual problem to find a balanced vector $\vec{x}$. The dual problem can be solved by an iterative procedure where each iteration maximizes the smallest slack as described in the following. 

We denote by $s_i$ the slack of  node $i$ defined by $$s_i = x_i - \max_{(l,i)\in E\setminus M}(w_{i,l} - x_l)_+$$ while the slack of edge $(i,j)$  denoted by $s_{i,j}$ is defined by  $$s_{i,j} = x_i+x_j-w_{i,j}\:.$$ Indeed, for a stable outcome $\vec{x}$, $s_{i,j}\geq 0$, for every $(i,j)\in E$. It is not difficult to observe that node and edge slacks satisfy $$s_i = \min(x_i,\min_{(i,l)\in E\setminus M}s_{i,l})\:.$$

The KT procedure for finding a balanced outcome proceeds by successively fixing the values $x_i$ for some nodes in $V$. This is allowed by the following key property~\cite{KT08}: if there exists a set $A\subset V$ and $\sigma \geq 0$ such that $s_i \leq \sigma$ for every $i\in A$ and $s_i \geq \sigma$ for $i\in V\setminus A$ and a vector $\vec{x}$ such that values $x_i$ are balanced in $A$, then there exists a vector $\vec{x}'$ such that $x_i' = x_i$ for every $i\in A$ that is a balanced outcome for $G$. 

The KT algorithm is sketched as follows. Let $\sigma \geq 0$ be a variable and let $A$ be a set of nodes for which values $x_i$ have been already assigned. The set $A$ is constructed such that no matched edge crosses the cut $(A,V\setminus A)$, i.e. for every node $i\in A$ there exists no node $j\in V\setminus A$ such that $(i,j)\in M$. Initially, $\sigma = 0$ and the set $A$ contains all the unmatched nodes. The algorithm then proceeds inductively with respect to the number of nodes with unassigned values as given by $|V\setminus A|$. Given $\sigma$ and $A$ the inductive step amounts to assigning values to nodes in $V\setminus A$ that maximize the minimum slack $\sigma' \geq \sigma$ which amounts to solving the following linear program 
\begin{equation}\label{eq-KTLP}
\begin{array}{rl}
\hbox{maximize} & \sigma'\\
\hbox{subject to} & x'_i \geq \sigma',\ i \in V\setminus A\\
& x'_i + x'_j = w_{i,j},\ (i,j)\in M\\
& x'_i + x'_j \geq w_{i,j} + \sigma',\ (i,j) \in E\setminus (M\cup E(A))\\
& x'_i = x_i, i \in A,
\end{array}\end{equation}
where $E(A)$ corresponds to the set of edges of the graph $G$ linking nodes in $A$.

For a fixed $\sigma'$, this is a linear inequalities' problem with at most two variables per inequality, for which polynomial algorithms exist. In particular, by results of Aspvall and Shilach~\cite{AS80}, for a given $\sigma'$, the system of inequalities is infeasible if there exists an infeasible simple loop in the graph construction described in~\cite{AS80}. A path is said to be a loop if the initial and final nodes are identical and is said to be simple if all intermediate nodes of this path are distinct. Furthermore, if a feasible solution exists than it can be constructed by finding the most constraining feasible simple loop. For the above system of inequalities, any such feasible simple loop is either a path, a cycle, a blossom or a bicycle. We refer to these as \emph{KT elementary graphs} and define them in the following:
\begin{itemize}
\item {\bf Path}. A path consists of alternating matchings with each of its end nodes anchored at either a node $i\in A$ or at a matched edge $(i,j)\in M$ such that $s_j = x_j$.\footnote{Recall that if for a matched edge $(i,j)\in E$, $i\in V\setminus A$, then also $j\in V\setminus A$, and vice versa.}
\item {\bf Cycle}. A cycle consists of an even number of nodes connected by a path of alternating matchings. 
\item {\bf Blossom}. A blossom is a concatenation of a cycle and a path (we refer to it as a stem) as follows. The cycle consists of an odd number of nodes that are connected by a cycle of alternating matchings started from a node (we call gateway) with an unmatched edge. The stem is a path of alternating matchings such that one end node is matched to the gateway node and the other end node is anchored as for a path.
\item {\bf Bicycle}. A bicycle is a concatenation of two blossoms by connecting the end nodes of their respective stems such that the connected stems form alternating matchings. 
\end{itemize}

The above described step is repeated until all the nodes are assigned values, i.e. until $V\setminus A = \emptyset$. Hence, the total number of such steps $k$ is at most the number of nodes $|V|$. At each step $l$, a KT elementary structure $C_l$ and maximum slack $\sigma_l$ are identified such that the $\sigma_l$'s form an non-decreasing sequence, $0 = \sigma_0 \leq \sigma_1 \leq \cdots \leq \sigma_k$. 

%{\bf A positive gap condition}. A balanced outcome $\vec{x}$ with slacks $\sigma_0, \sigma_1,\ldots,\sigma_k$ is said to have a gap $\sigma > 0$ if, for every $1\leq l \leq k$,
%$$
%\sigma_l - \sigma_{l-1} \geq \sigma
%$$
%and, for every pair of nodes $i$ and $j$ of $C_l$ such that the edge $(i,j)$ is not part of $C_l$, we have
%$$
%x_i + x_j - w_{i,j} \geq \sigma_l + \sigma.
%$$
%
%This positive gap condition ensures 
%
% enables to study convergence of a dynamic process by partitioning into a sequence of KT elementary graphs and decoupling the dynamics over these elementary graphs, a technique introduced and used in~\cite{KBBCM10}. 

\subsection{Convergence}

We introduce a few elementary concepts about stability of dynamical systems in a somewhat informal manner and then define the notion of convergence time considered in this paper. We say that a dynamical system, according to which $\vec{x}(t)$ evolves over $t\geq 0$, is \emph{asymptotically stable}, if for every initial value $\vec{x}(0)$, there exists a point $\vec{x}^*$ such that 
$$
\lim_{t\rightarrow \infty}||\vec{x}(t)-\vec{x}^*|| = 0.
$$
The system is said to be \emph{globally asymptotically stable} if $\vec{x}^*$ is unique, i.e. does not depend on the initial value $\vec{x}(0)$. 

In particular, for a linear system $$\vec{x}(t+1)=\matA \vec{x}(t) +\vec{b}(t)$$ where $\matA$ is a given matrix and $\vec{b}(t)$ is a vector that may depend on $t$, we have that the system is globally asymptotically stable if the spectral radius of the matrix $\matA$ is smaller than $1$ (i.e. all eigenvalues are of modulo strictly smaller than $1$). The concepts of asymptotic stability and global asymptotic stability are standard, see~\cite{K01} for more details.

We say that the convergence to a point $\vec{x}^*$ is exponentially bounded if there exist $C > 0$ and $R>0$ such that for every initial value $\vec{x}(0)$, we have
$$
||\vec{x}(t) - \vec{x}^*|| \leq C e^{-R t}, \hbox{ for every } t \geq 0,
$$
where we refer to $R$ as the rate of convergence and call $T = 1/R$ the convergence time. Moreover, If $\vec{x}(t)$ evolves according to the aforementioned linear system then the rate of convergence is given by $(i)$ $R = \log(1/\rho(\matA))$ where $\rho(\matA)$ is the spectral radius of matrix $\matA$ if the system is globally asymptotically stable, and $(ii)$ $R = \log(1/\lambda_2(\matA))$ where $\lambda_2(\matA)$ is the modulus of the largest eigenvalue of matrix $\matA$ that is smaller than $1$, if the system is asymptotically stable.

\section{Edge-Balanced Dynamics for KT Elementary Graphs}
\label{sec:KTsub}

The {\em edge-balanced dynamics }was first considered by Rochford~\cite{R84} and Cook and Yamagishi~\cite{CY92}, this dynamical process assumes that players already agreed on a matching $M$ and are negotiating the value of the outcome $\vec{x}$. Hence, each matched player $i$ is assigned a trading partner, which we denote with $p_i$. A version of this dynamics in discrete-time can be represented as follows. For a fixed $0<\alpha \leq 1$ and given an initial value $\vec{x}(0)$, for $i = 1,2,\ldots,n$ and $t = 0,1,\ldots$, we have that
\begin{equation}
x_i(t+1) = x_i(t) + \alpha \left\{\left[y_i(t)+\frac{1}{2}(w_{i,p_i}-y_i(t)-y_{p_i}(t))\right]_0^{w_{i,p_i}}-x_i(t)\right\}
\label{equ:edgebalanced}
\end{equation}
where $y_l(t)$ is the best alternate value that a matched player $l$ may get at time $t$ by trading with her other neighbors, i.e.
$$
y_l(t) = \max_{k: (l,k)\in E\setminus M} (w_{l,k} - x_k(t))_+
$$ 
and we use the notation $[\cdot]_a^b = \min(\max(\cdot,a),b)$, for $a\leq b$. 

It is not difficult to observe that if players $i$ and $j$ are matched, then $x_i(t) + x_{j}(t) = w_{i,j}$ is time invariant, i.e. if the latter holds for a time $t$, then it still holds for time $t+1$. Note that the dynamics is not necessarily consistent with Nash bargaining solution for every time $t$ as for a matched pair $(i,j)$, the edge-surplus $w_{i,j}-y_i(t)-y_j(t)$ is allowed to be negative; the only requirement is that the allocation $y_{i}(t) + \frac{1}{2} (w_{i,j}-y_i(t)-y_j(t))$ is in $[0,w_{i,j}]$. However, the edge surpluses are guaranteed to be positive for $t$ large enough~\cite{ABCDP09}. Besides the dynamics described in (\ref{equ:edgebalanced}) is guaranteed to converge to a fixed point that corresponds to a Nash bargaining solution, see~\cite[Theorems 1,2]{ABCDP09}.

In this section, we will observe that for all the elementary graphs of the KT decomposition, the values held by the nodes eventually evolve according to a \emph{linear} discrete-time dynamical system, i.e., for a given matrix $\matA$ and a vector $\vec{b}(t)$, $\vec{x}(t)$ evolves according to
\begin{equation}
\vec{x}(t+1) = \matA \vec{x}(t) + \vec{b}(t).
\label{equ:lin}
\end{equation}
We will find that for a path and a cycle the dynamics is linear for every time $t\geq 0$ while for a blossom and a bicycle there exists a finite time $T_0\geq 0$ such that the dynamics is linear for every $t\geq T_0$. The asymptotic behavior is determined by spectral properties of matrix $\matA$. Note that it suffices to consider the spectrum of matrix $\matA$ for $\alpha = 1$. This is because, for every given $0<\alpha\leq 1$, $\lambda' = 1 - \alpha + \alpha \lambda$ is an eigenvalue and $\vec{v}$ is an eigenvector of the matrix $\matA$, where $\lambda$ is an eigenvalue and $\vec{v}$ is an eigenvector of the matrix $\matA$ under $\alpha = 1$. We will see that for every KT elementary graph, the eigenvalues of matrix $\matA$, under $\alpha = 1$, are located in the interval $[-1,1]$ and will show that $-1$ can be an eigenvalue only for a cycle with an even number of matched edges or a bicycle with an even number of matched edges in each of its loops. In the latter two cases, for $\alpha = 1$, there is no convergence to a limit point as the asymptotic behavior is periodic because of the eigenvalue $-1$. This is ruled out by choosing the smoothing parameter $0<\alpha < 1$, making all the eigenvalues strictly larger than $-1$, and thus ensuring convergence to a limit point.

It is also worth noting that if the system is globally asymptotically stable and if when, the term $b(t)$ in (\ref {equ:lin}) is equal to $b$ a constant, that does not depend on t, then the fixed point to which the system converges is given by $(I-A)^{-1} b$ where $(I-A)^{-1}=\sum_{k\geq 0} A^k$. We would like to mention that the analysis of the fixed points of these dynamical systems is relevant to understanding the bargaining power of nodes depending on their position in the network.

Finally, we note that for our results in this section, we assume uniform edge weights and under this assumption, without loss of generality, we let $w_e = 1$, for every $e\in E$. 

\subsection{Path}
\label{sec:path}

We consider a path with $n$ matched edges with boundary values $x^+(t)$ and $x^-(t)$ as illustrated in Figure~\ref{fig:path}. In this case, the evolution of the node values $\vec{x}(t)$ boils down to a discrete-time linear dynamical system (\ref{equ:lin}) where $\matA$ is the $n\times n$ \emph{symmetric tridiagonal} matrix 
\begin{equation}
\matA=\left( \begin{array}{ccccc}
0 & 1/2 & 0& \cdots &0\\
1/2 & \ddots  & \ddots &\ddots &\vdots\\
0 & \ddots &\ddots &\ddots& 0\\
\vdots & \ddots & \ddots & \ddots & 1/2\\
0&\cdots& 0& 1/2 & 0 \end{array} \right)\:
\label{equ:Atridiagonal}
\end{equation}
and $\vec{b}(t) = (\frac{1-x^+(t)}{2},\underbrace{0,\ldots,0}_{n-2},\frac{x^-(t)}{2})^T$.

\begin{figure}[ht!]
\begin{center}
\vspace*{5mm}
\psfrag{c1}{$x^+$}
\psfrag{c2}{$x^-$}
\psfrag{x1}{$x_1$}
\psfrag{x2}{$x_2$}
\psfrag{xn}{$x_n$}
\psfrag{cdots}{$\cdots$}
\includegraphics[width=3in]{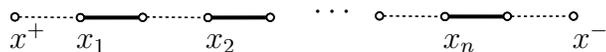}
\vspace*{-2mm}
\caption{A path with boundary conditions. Note that the values at the ends of each matched edge are given by $x_i$ and $1-x_i$, the latter being omitted in the figure.}
\vspace*{5mm}
\label{fig:path}
\end{center}
\end{figure}

The eigenvalues of matrix $\matA$ are
$
\lambda_k = \cos\left(\frac{\pi k}{n+1}\right),\ $ $k = 1,2,\ldots,n, 
$
with the corresponding orthonormal eigenvectors 
$$
\vec{v}_{k} = \sqrt{\frac{2}{n+1}}\left(\sin\left(\frac{\pi k}{n+1}\right),\ldots,\sin\left(\frac{\pi k n}{n+1}\right)\right)^T.
$$
Note that every eigenvalue is of modulo smaller than $1$. This implies asymptotic stability for every $0 < \alpha \leq 1$. From the above spectrum, we have the following characterization of the convergence time:

\begin{theorem} For a path of $n$ matched edges and every $0 < \alpha \leq 1$, the convergence time is 
$$
T = \frac{2}{\alpha \pi^2} n^2 \cdot [1 + O(1/n^2)].
$$
\label{thm:pathtime} 
\end{theorem}
\begin{proof}

An eigenvalue $\lambda$ and associated eigenvector $\vec{v}$ of matrix $\matA$ satisfy
\begin{eqnarray*}
\lambda v_1 &=& \frac{1}{2}v_2\\
\lambda v_i &=& \frac{1}{2}v_{i-1} + \frac{1}{2}v_{i+1},\ 1 < i < n\\
\lambda v_n &=& \frac{1}{2}v_{n-1}.
\end{eqnarray*}
Using $\lambda = \cos(\phi)$ and $v_i = \sin(\phi i)$, for $\phi\geq 0$ in the above equations, along with some elementary trigonometric calculus, it readily follows that $\phi = \frac{\pi k}{n+1}$, for $k = 1,2,\ldots,n$.

Since $-1 < \lambda_k < 1$ for every $k$ and $\lambda_1 > 0$ has the largest modulo, the convergence time is given by $T = \log(1-\alpha + \alpha \lambda_1)$. Noting that $\lambda_1 = 1 - \frac{\pi^2}{2n^2} + O(1/n^4)$, the asserted result follows.
\qfd
\end{proof}

From this theorem, we observe that the convergence time is quadratic in the number of matched edges. Moreover, if $x_+=x_-=0$ then it is not difficult to see that the dynamics convergence to the fixed point given by $$(I-A)^{-1}(1/2,0,\ldots,0)^T.$$

\subsection{Cycle}

For an alternating cycle between $n$ nodes ($n$ even), the dynamics of node values $\vec{x}(t)$ boils down to a linear dynamical system (\ref{equ:lin}) where $\matA$ is the following \emph{circulant} matrix, for $n = 2$, $\matA = \left(\begin{array}{cc}0 & 1\\ 1 & 0\end{array}\right)$, and otherwise  
\begin{equation}
\matA=\left( \begin{array}{ccccccc}
0 & 1/2 & 0 & \cdots & 0 & 0 & 1/2\\
1/2 & 0 & 1/2 & \ddots &\ddots & \ddots & 0\\
0 & 1/2 & \ddots & \ddots & \ddots & \ddots & 0\\
\vdots & \ddots & \ddots & \ddots & \ddots & \ddots & \vdots\\
0 & \ddots & \ddots & \ddots & \ddots & 1/2 & 0\\
0 & \ddots & \ddots & \ddots & 1/2 & 0 & 1/2\\
1/2 & 0 & 0 & \cdots & 0 & 1/2 & 0 
\end{array} \right)\:
\label{equ:Acycle}
\end{equation}
and where vector $\vec{b} = \vec{0}$. Note that in this case 
$
\vec{x}(t) = \matA^t \vec{x}(0),\ \hbox{for } t \geq 0.
$

\begin{figure}[htbp]
\centering
\vspace*{5mm}
\psfrag{x1}{$x_1$}
\psfrag{x2}{$x_2$}
\psfrag{xn}{$x_n$}
\includegraphics[width=1in]{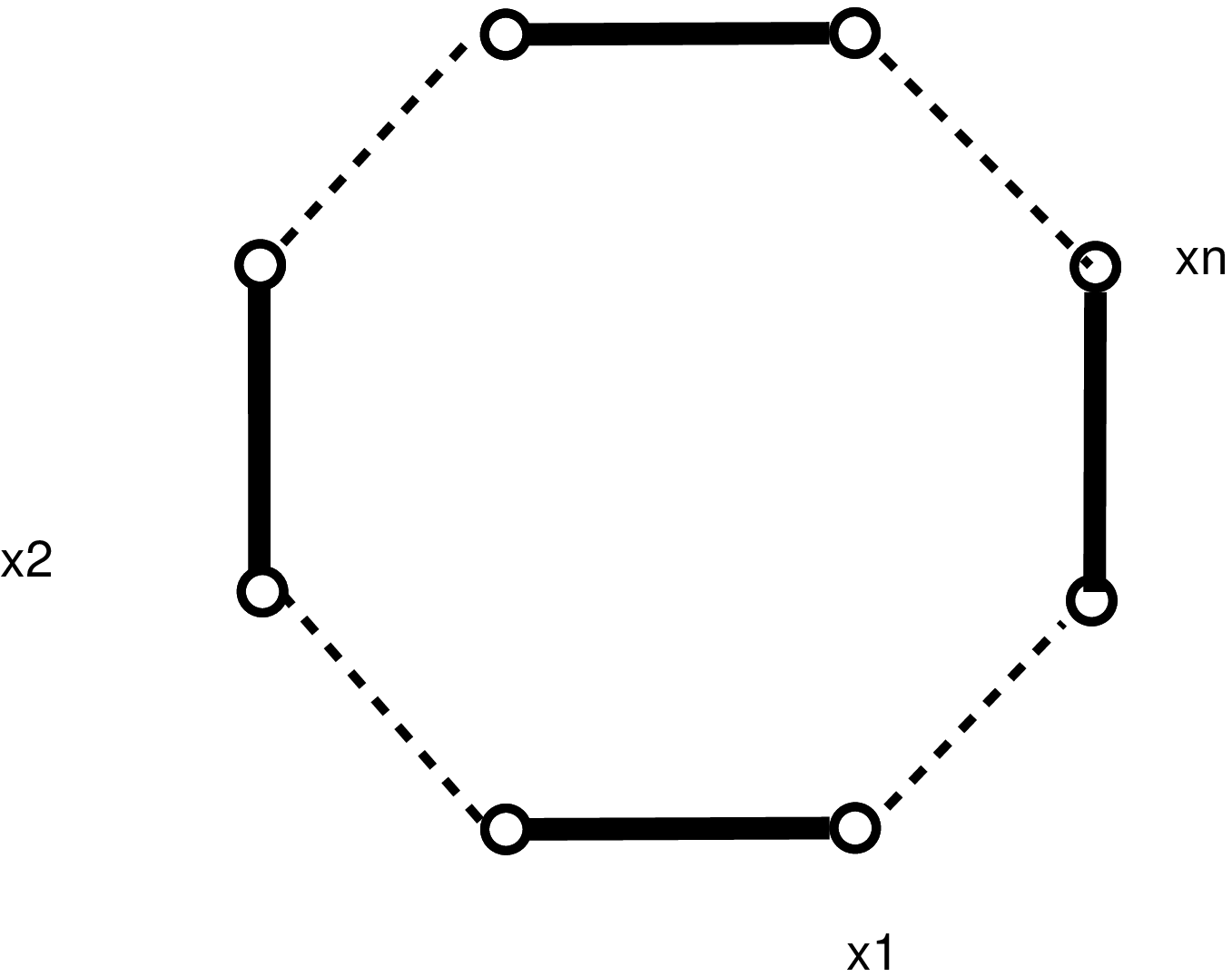}
\vspace*{5mm}
\caption{A cycle.}
\label{fig:cycle}
\end{figure}

By using similar arguments as for a path, it is not difficult to establish that the eigenvalues of matrix $\matA$ are 
$
\lambda_k = \cos\left(\frac{2\pi(k-1)}{n}\right), $ $k = 1,2,\ldots,n,
$
with the corresponding orthonormal eigenvectors
$$
\vec{v}_k = 
\left\{
\begin{array}{l}
\frac{1}{\sqrt{n}}\left(1, 1,\ldots, 1, 1\right)^T, \hbox{ if } k = 1\\
\frac{1}{\sqrt{n}}\left(1, -1,\ldots, 1,-1\right)^T, \hbox{ if } k = 1 + n/2\\
\sqrt{\frac{2}{n}}\left(1, \cos\left(\phi_k\right),\ldots, \cos\left(\phi_k(n-1)\right)\right)^T, \hbox{ o.w.}
\end{array}
\right .
$$
where for ease of notation, $\phi_k = \frac{2\pi(k-1)(n-1)}{n}$.

Using the spectral decomposition of the symmetric matrix $\matA$ (see \cite{K01} for details), we have
\begin{equation}
\vec{x}(t) = \sum_{k=1}^n \lambda_k^t \vec{v}_k \vec{v}_k^T\vec{x}(0).
\label{equ:xspec}
\end{equation}
%For $n$ is even, $\lambda_k = -1$, for $k = 1 + n/2$, and $\lambda_k > -1$, otherwise. From (\ref{equ:xspec}), we have
%$$
%\vec{x}(t) = \left(\vec{v}_1 \vec{v}_1^T + (-1)^t \vec{v}_{1+n/2}\vec{v}_{1+n/2}^T\right)\vec{x}(0) + o(1).
%$$ 
%Therefore, the asymptotic behavior is periodic.
We distinguish two cases:
\begin{itemize}
\item {\bf Case 1}: $n$ is even. In this case, $\lambda_k = -1$, for $k = 1 + n/2$, and $\lambda_k > -1$, otherwise. From (\ref{equ:xspec}), we have
$$
\vec{x}(t) = \left(\vec{v}_1 \vec{v}_1^T + (-1)^t \vec{v}_{1+n/2}\vec{v}_{1+n/2}^T\right)\vec{x}(0) + o(1).
$$ 
Therefore, the asymptotic behavior is periodic.

\item {\bf Case 2}: $n$ is odd. In this case, $-1 < \lambda_k \leq 1$, for every $k$, and thus we have asymptotic convergence to the limit point, $\lim_{t\rightarrow \infty}x_i(t) = \frac{1}{n}\sum_{j=1}^n x_j(0)$, for every $i$. 
\end{itemize}

In view of the above observations, we note that for even $n$, we need to assume that $\alpha$ is strictly smaller than $1$ in order to rule out asymptotically periodic behavior, while for odd $n$, we can allow for $\alpha = 1$. Moreover, it is worth noting that the limit point, in both cases, depends on the initial condition and is given by 
$ \frac{1}{n}\sum_{j=1}^n x_j(0) (1,\ldots,1)^T.$

The following result shows that in like manner as for a path, the convergence time is quadratic in the number of matched edges, but note that it is four times smaller, asymptotically for large $n$. 

\begin{theorem} For cycle graph with $n$ matched edges  and  for $\alpha\in(0,1) $  the convergence time is 
$$
T = \frac{1}{\alpha 2\pi^2} n^2 \cdot [1 + O(1/n^2)].
$$ 
\end{theorem}
%The theorem readily follows from the above asserted spectrum of matrix $\matA$ and noting that the largest modulo of an eigenvalue is $1 - \alpha +\alpha\lambda_2$ and $\lambda_2 = 1 - \frac{2\pi^2}{n^2} + O(1/n^4)$.

\subsection{Blossom}
\label{sec:blossom}

A blossom is a concatenation of a cycle and a path (we refer to it as a stem); see Figure~\ref{fig:blossom} for an example. We consider a blossom with $n$ matched edges in the stem and $m$ matched edges in the loop. We refer to the node that connects the stem and the loop as a \emph{gateway} node. The matched edges of the stem are enumerated as $1,2,\ldots,n$ along the stem towards the gateway note. We let $x_i$ denote the value of the end node of an edge $i$ of the stem that is connected to a node towards the open end of the stem. Similarly, we enumerate matched edges of the loop as $1,2,\ldots,m$ and let $y_i$ denote the value of the node that appears first on a matched edge $i$ as we go along the loop in the clockwise direction.

\begin{figure}[htbp]
\vspace*{5mm}
\centering
\psfrag{x1}{$x_1$}
\psfrag{x2}{$x_2$}
\psfrag{xn}{$x_n$}
\psfrag{y1}{$y_1$}
\psfrag{y2}{$y_2$}
\psfrag{ym}{$y_m$}
\includegraphics[width=3in]{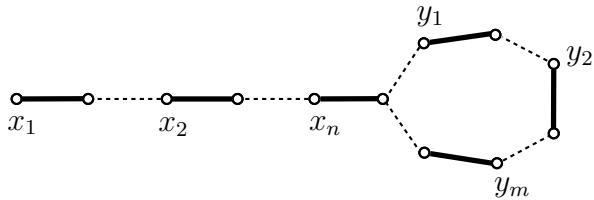}
\vspace*{5mm}
\caption{A blossom.}
\label{fig:blossom}
\end{figure}

It can be observed that node values $\vec{x}(t)$ and $\vec{y}(t)$ evolve according to the following non-linear dynamical system:
\begin{equation}
\begin{split}
x_1(t+1) &= \frac{x_2(t)}{2}\\
x_i(t+1) &= \frac{x_{i-1}(t) + x_{i+1}(t)}{2},\quad 1 < i < n\\
x_n(t+1) &= \frac{1+x_{n-1}(t) - \max[1-y_1(t),y_m(t)]}{2}\\
y_1(t+1) &= \frac{x_n(t) + y_{2}(t)}{2}\\
y_i(t+1) &= \frac{y_{i-1}(t) + y_{i+1}(t)}{2}, \quad 1 < i < m\\
y_m(t+1) &= \frac{1 + y_{m-1}(t) - x_n(t)}{2}.
\end{split}
\label{equ:blossomdyn}
\end{equation}

Note that the system is non-linear only because of the maximum operator that acts in the update for the node matched to the gateway node, which connects the stem and the loop of the blossom. The maximum operator is over the values of the nodes that are in the loop and are matched to the neighbors of the gateway node, $1-y_1(t)$ and $y_m(t)$. It turns out that, eventually, one of these two values is larger or equal to the other and, hence, the system dynamics becomes linear. This is shown in the following lemma. Note that the sum $y_1(t) + y_m(t)$, if smaller than or equal to $1$, indicates $\max(1-y_1(t),y_m(t)) = 1-y_1(t)$, and otherwise, $\max(1-y_1(t),y_m(t)) = y_m(t)$. 

\begin{theorem} For a blossom with $n$ matched edges in the stem and $m$ matched edges in the loop, for every initial value $(\vec{x}(0),\vec{y}(0))$, the sum of node values $y_1(t) + y_m(t)$ satisfies:
\begin{enumerate} 
\item $y_1(t) + y_m(t)$, for $t\geq 0$, is autonomous of $\vec{x}(t)$, $t\geq 0$.
\item $\lim_{t\rightarrow \infty} y_1(t) + y_m(t) = 1$.
\item The asymptotic rate of convergence is $\frac{\pi^2}{2m^2}$.
\item There exists a time $T_0\geq 0$ such that either $y_1(t)+y_m(t) \leq 1$ or $y_1(t)+y_m(t) \geq 1$ for every $t \geq T_0$.
\item $T_0 = O(m^2)$.
\end{enumerate}
\label{lem:y1m}
\end{theorem}

As an aside remark, note that the value sum $\sum_{i=1}^m y_i(t)$ for the loop nodes evolves autonomously from $\vec{x}(t)$, $t\geq 0$. To see this, from (\ref{equ:blossomdyn}) note
$$
\sum_{i=1}^m y_i(t+1) = \sum_{i=1}^m y_i(t) - \frac{1}{2}(y_1(t)+y_m(t)) + \frac{1}{2}.
$$
and that the claim follows from Theorem~\ref{lem:y1m}~item~1, saying that $y_1(t)+y_m(t)$ evolves autonomously of $\vec{x}(t)$, $t\geq 0$.

The theorem derives from an explicit characterization of $y_1(t)+y_m(t)$, for every $t\geq 0$, which we present in the following:

\begin{lemma} Given initial value $\vec{y}(0)$, for every $t\geq 0$,
\begin{equation}
y_1(t) + y_m(t) = 1 - \frac{2}{m+1}\sum_{i=1}^{\lceil\frac{m}{2}\rceil} f_{2i-1}(\vec{y}(0)) \lambda_{2i-1}^t 
\label{equ:sumy}
\end{equation}
where
$
f_k(\vec{y}) = 1+\lambda_k - 2\sqrt{1-\lambda_k^2}\sqrt{\frac{m+1}{2}}\vec{v}_k^T\vec{y}.
$
\label{lem:sumy}
\end{lemma}
\begin{proof}

The part of the system $\vec{y}(t)$ evolves as the following non-autonomous linear system
$
\vec{y}(t+1) = \matA \vec{y}(t) + \vec{b}(t)
$
where $\matA$ is a tridiagonal matrix that corresponds to a path of $m$ matched edges and $\vec{b}(t) = (x_n(t)/2,\underbrace{0,\ldots,0}_{m-2},(1-x_n(t))/2)^T$. 
 
Since $\matA$ is a symmetric matrix, we can use the spectral decomposition 
$
\matA = \sum_{k=1}^m \lambda_k \vec{v}_k \vec{v}_k^T
$
where $\lambda_1,\lambda_2,\ldots,\lambda_m$ are the eigenvalues and $\vec{v}_1,\vec{v}_2,\ldots,\vec{v}_m$ are the orthonormal eigenvectors of matrix $\matA$, which we identified in Section~\ref{sec:path}.

Using the spectral decomposition, we note
$$
y_i(t) = \sum_{k=1}^m \lambda_k^t v_{k,i} \vec{v}_k^T \vec{y}(0) + \sum_{s=0}^{t-1} \lambda_k^{t-s-1} v_{k,i} \vec{v}_k^T \vec{b}(s)
$$
where $v_{k,i} =\sqrt{\frac{2}{m+1}} \sin\left(\frac{\pi k}{m+1}i\right)$ is the $i$-th coordinate of the eigenvector $\vec{v}_k$. Summing up $y_1(t)$ and $y_m(t)$, we obtain
\begin{eqnarray*}
 y_1(t) + y_m(t)
&=& \sum_{k=1}^m (v_{k,1}+v_{k,m}) \left(\lambda_k^t \vec{v}_k^T \vec{y}(0) + \sum_{s=0}^{t-1} \lambda_k^{t-s-1} \vec{v}_k^T \vec{b}(s)\right)\\
&=& \sum_{k\ \hbox{odd}} 2v_{k,1} \left(\lambda_k^t \vec{v}_k^T \vec{y}(0) + \sum_{s=0}^{t-1} \lambda_k^{t-s-1} \vec{v}_k^T \vec{b}(s)\right)
\end{eqnarray*}
where the last inequality is because of the fact $v_{k,m} = v_{k,1}$ for $k$ odd and $v_{k,m} = - v_{k,1}$ for $k$ even. Furthermore, 
\begin{eqnarray*}
\vec{v}_k^T \vec{b}(s) &=& v_{k,1} \frac{x_n(s)}{2} + v_{k,m} \frac{1-x_{n}(s)}{2} \frac{v_{k,1}}{2} \hbox{ for } k \hbox{ odd}.
\end{eqnarray*}
Therefore, 
\begin{eqnarray*}
y_1(t) + y_m(t)
&=& \sum_{k\ \hbox{odd}} \left(\lambda_k^t 2v_{k,1}\vec{v}_k^T \vec{y}(0) + \sum_{s=0}^{t-1} \lambda_k^{t-s-1} v_{k,1}^2\right)\\
&=& \sum_{k\ \hbox{odd}} \left(\lambda_k^t 2v_{k,1}\vec{v}_k^T \vec{y}(0) + \frac{1-\lambda_k^t}{1-\lambda_k} v_{k,1}^2\right)\\
&=& \sum_{k\ \hbox{odd}} \left(\lambda_k^t 2\sqrt{\frac{2}{m+1}}\sqrt{1-\lambda_k^2}\vec{v}_k^T \vec{y}(0)  +(1-\lambda_k^t)\frac{2}{m+1}(1+\lambda_k)\right)\\
&=& \frac{2}{m+1}\sum_{k\ \hbox{odd}} (1+\lambda_k) -\frac{2}{m+1} \sum_{k\ \hbox{odd}} \left(1+\lambda_k - 2\sqrt{\frac{m+1}{2}}\sqrt{1-\lambda_k^2}\vec{v}_k^T \vec{y}(0)\right)\lambda_k^t.
\end{eqnarray*}
It remains only to show that 
$$
\frac{2}{m+1}\sum_{k\ \hbox{odd}} (1+\lambda_k) = 1
$$
which follows readily by elementary trigonometric calculations.
\qfd

\end{proof}

We now give the proof of Thereom \ref{lem:y1m}.
\begin{proof}[Proof of Thereom \ref{lem:y1m}.]

The statements of the theorem derive from Lemma~\ref{lem:sumy} as follows. Item 1 clearly holds as the function (\ref{equ:sumy}) depends only on the initial value $\vec{y}(0)$. Item~2 follows from (\ref{equ:sumy}) because all the eigenvalues $\lambda_k$ are real and with modulo strictly smaller than $1$. Item~3 holds from the fact that the largest modulo eigenvalue of matrix $\matA$ is $\lambda_1 = \cos\left(\frac{\pi}{m+1}\right) = 1 - \frac{\pi^2}{2m^2} + O(1/m^4)$ and hence $R = \log(1/\lambda_1) = \frac{\pi^2}{2m^2} + O(1/m^4)$. Item~4 holds as the sum in (\ref{equ:sumy}) is asymptotically dominated by the largest modulo eigenvalue $\lambda_{2i-1}$ such that $\vec{v}_{2i-1}^T\vec{y}(0) \neq 0$, i.e. the mode associated to the eigenvalue $\lambda_{2i-1}$ is excited. Let us consider the case where such an eigenvalue is $\lambda_1$ and $m$ is even; the other cases follow by similar arguments. From Lemma~\ref{lem:sumy}, we have $y_1(t) + y_m(t)=$
$$1 - \frac{2}{m+1} \lambda_1^t    \left(f_1(\vec{y}(0)) + \sum_{i=2}^{\lceil\frac{m}{2}\rceil}f_{2i-1}(\vec{y}(0))\left(\frac{\lambda_{2i-1}}{\lambda_1}\right)^t\right),$$
and, thus, since $|\lambda_{2i-1}/\lambda_1| < 1$, for every $1<i\leq \lceil m/2\rceil$, for $t$ large,
$$
y_1(t) + y_m(t) = 1 - \frac{2}{m+1} \lambda_1^t \left[f_1(\vec{y}(0)) +  o(1)\right].
$$
Finally, item~5 holds as, for $m$ large enough,
$$
\gamma: =\max_i |\frac{\lambda_{2i-1}}{\lambda_1}|\leq \frac{\lambda_3}{\lambda_1} = 1 - \frac{4\pi^2}{m^2} + O(1/m^4).
$$ 
For an arbitrary $\epsilon > 0$, we have $|\lambda_{2i-1}/\lambda_1|^t \leq \epsilon$, for every $i > 1$, provided that time $t$ is such that
$$
t \geq \frac{\log\left(\frac{1}{\epsilon}\right)}{\log\left(\frac{1}{\gamma}\right)} =  \frac{\log\left(\frac{1}{\epsilon}\right)}{4\pi^2} m^2 [1+o(1)].
$$
Hence, $T_0 = O(m^2)$.

\qfd

\end{proof}

From Theorem~\ref{lem:y1m}~item~4, we have that the dynamics for a blossom is eventually according to the following linear system
$$
\left(\begin{array}{c}\vec{x}(t+1)\\\vec{y}(t+1)\end{array}\right) = \matA \left(\begin{array}{c}\vec{x}(t)\\\vec{y}(t)\end{array}\right) + \vec{b}
$$
where matrix $\matA$ and vector $\vec{b}$ assume one of the following two choices: 

\begin{itemize}
\item {\bf Case~1}: $(1-y_1(t)\geq y_m(t))$ 
\begin{equation}
\matA=
\left(
\begin{array}{cc}
\matT_n & \matP\\
\matQ & \matT_m
\end{array} 
\right)
\label{equ:Ablossom}
\end{equation}
with $\matT_n$ and $\matT_m$ tridiagonal matrices of paths of $n$ and $m$ matched edges, respectively, and 
$$
\hspace*{-5mm}\matP=\left(\begin{array}{cccc}
0 & 0 & \cdots & 0\\
0 & 0 & \cdots & 0\\
\vdots & \vdots & \cdots & \vdots\\
0 & 0 & \cdots & 0\\
\frac{1}{2} & 0 & \cdots & 0 
\end{array}\right)\:, 
\matQ=\left(\begin{array}{cccc}
0 & \cdots & 0 & \frac{1}{2}\\
0 & \cdots & 0 & 0\\
\vdots & \cdots &  \vdots & \vdots\\
0 & \cdots & 0 & 0\\
0 & \cdots & 0 & -\frac{1}{2} 
\end{array}\right).
$$
and $\vec{b} = (\underbrace{0,\ldots,0}_{n+m-1},1/2)^T$.

Moreover, if the initial condition is such that $1-y_1(0)\geq y_m(0)$, e.g. $y_i(t)=0$ for all $i$, then the fixed point is given by $(I-A)^{-1}\vec{b}. $
\item {\bf Case~2}: $(1-y_1(t) < y_m(t))$ same as under Case~1 but
$$
\matP=\left(\begin{array}{cccc}
0 & \cdots & 0 & 0\\
0 & \cdots & 0 & 0\\
\vdots & \cdots & \vdots & \vdots\\
0 & \cdots & 0 & 0\\
0 & \cdots & 0 & -\frac{1}{2} 
\end{array}\right)
$$
and $\vec{b} = (\underbrace{0,\ldots,0}_{n-1},1/2,\underbrace{0,\ldots,0}_{m-1},1/2)^T$.
Finally, it the initial condition is such that $1-y_1(0)< y_m(0)$, e.g. $y_i(t)=1$ for all $i$, then the fixed point is given by $(I-A)^{-1}\vec{b}. $
\end{itemize}

In the following we only consider Case~1 as the spectrum of matrix $\matA$ under Case~2 is exactly the same. We note that the eigenvalues of the matrix $\matA$ are $(\lambda_1,\lambda_2,\ldots,\lambda_{n+\lfloor m/2\rfloor},\mu_1,\mu_2,\ldots,\mu_{\lceil m/2\rceil})$ where
\begin{eqnarray*}
\lambda_k &=& \cos\left(\frac{2\pi k}{2n+m+1}\right),\ k = 1,\ldots, n + \lfloor m/2 \rfloor\\
\mu_k &=& \cos\left(\frac{\pi (2k-1)}{m+1}\right),\ k = 1,\ldots, \lceil m/2 \rceil
\end{eqnarray*}
with a proof provided in Section \ref{sec:eigbloss}.

It is noteworthy that all the eigenvalues have modulo strictly smaller and $1$, and thus, the system is globally asymptotically stable. We now characterize the convergence time from an instance at which the system became linear.

\begin{theorem}\label{th_blossom} For a blossom with $n$ matched edges in the stem and $m$ matched edges in the loop, for every $0 < \alpha \leq 1$, the convergence time $T$ satisfies: if $m$ is even, then
$$
T = \frac{2}{\alpha\pi^2}(2n+m)^2 \cdot [1 + o(1)]
$$
otherwise, for $m$ odd,
$$
T = \frac{2}{\alpha\pi^2} \max\left(m^2, \frac{1}{4}(2n + m)^2\right) \cdot [1 + o(1)].
$$
\end{theorem}

{\bf Observations}. The result implies that the convergence time is $O((n+m)^2)$, i.e. quadratic in the number of matched edges. There is a significant difference with regard to whether the number of matched edges in the loop, $m$, is even or odd. The convergence is slower for $m$ even. Specifically, if the length of the stem is at least twice the length of the loop, the convergence time is larger for a factor $4$. For a fixed $n$, the convergence time is asymptotically $\frac{2}{\alpha\pi^2} m^2$ as for a path of length $m$ which is intuitive. Likewise, if $m$ is fixed and odd, the convergence time is asymptotically $\frac{2}{\alpha\pi^2} n^2$ as for a path of length $n$ and thus also in conformance to intuition. 

\begin{proof} 

For the eigenvalues $\lambda_1,\lambda_2,\ldots, \lambda_{n+\lfloor m/2\rfloor}$, it is readily checked that 
$$
\max_k|\lambda_k| = -\lambda_{n+\lfloor m/2\rfloor} = \cos\left(\frac{(1 + 1_{m \hbox{ odd }})\pi}{2n+m+1}\right)
$$
where $1_{m \hbox{ odd }}$ stands for the indicator that $m$ is odd, while, on the other hand,
$$
\max_k |\mu_k| = \mu_1 = \cos\left(\frac{\pi}{m+1}\right).
$$
Therefore, the spectral radius of matrix $\matA$, $\rho(\matA) = \max(\max_k|\lambda_k|,\max_k |\mu_k|)$ is given by
$$
\rho(\matA) = \left\{
\begin{array}{ll}
\cos\left(\frac{\pi}{2n+m+1}\right), & m \hbox{ even }\\
\cos\left(\frac{2\pi}{2n+m+1}\right), & m \hbox{ odd }, m \leq 2n-1\\
\cos\left(\frac{\pi}{m+1}\right), & m \hbox{ odd }, m > 2n -1.
\end{array}
\right .
$$
The asserted asymptotic follows from the last identities.
\qfd
\end{proof}

\subsection{Bicycle}
\vspace*{-3mm}
\begin{figure}[htbp]
\centering
\psfrag{x1}{$x_1$}
\psfrag{x2}{$x_2$}
\psfrag{xn}{$x_n$}
\psfrag{y1}{$y_1$}
\psfrag{y2}{$y_2$}
\psfrag{ym}{$y_m$}
\psfrag{z1}{$z_1$}
\psfrag{z2}{$z_2$}
\psfrag{zl}{$z_l$}
\includegraphics[width=3in]{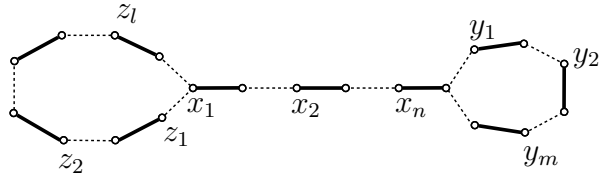}
\caption{A bicycle.}
\label{fig:bi-cycle}
\end{figure}

A bicycle graph consists of two loops that are connected by a path. Without loss of generality, we refer to one of the loops as loop 1 and to other as loop 2 and refer to the path as a cross-bar; see Figure~\ref{fig:bi-cycle} for an illustration. Notice that a bicycle graph corresponds to a concatenation of two blossoms by connecting the end nodes of their respective stems so that a cross-bar is formed of alternating matchings. We let $l$ and $m$ be the number of matched edges in loop 1 and loop 2, respectively, and let $n$ be the number of matched edges of the cross-bar. The values of the end nodes of the matched edges are denoted by $\vec{z}(t) = (z_1(t),z_2(t),\ldots,z_l(t))^T$,  $\vec{x}(t) = (x_1(t), x_2(t),\ldots, x_n(t))^T$ and $\vec{y}=(y_1(t),y_2(t),\ldots,y_m(t))^T$ for loop 1, cross-bar, and loop 2, respectively. See Figure~\ref{fig:bi-cycle} for positions of the corresponding nodes.

We note that for a bicycle the system evolves according to the following non-linear system:
\begin{equation}
\begin{split}
z_1(t+1) &= \frac{1+z_2(t)-x_1(t)}{2}\\
z_i(t+1) &= \frac{z_{i-1}(t) + z_{i+1}(t)}{2}, 1 < i < l\\
z_l(t+1) &= \frac{z_{l-1}(t) + x_1(t)}{2}\\ 
x_1(t+1) &= \frac{x_2(t) + \max[1-z_1(t), z_l(t)]}{2}\\
& \hbox { plus other updates as for blossom (\ref{equ:blossomdyn})}.
\end{split}
\label{equ:bicycledyn}
\end{equation}
In this case, the non-linearity originates because of two gateway nodes that connect the cross-bar with loops, each such gateway node having two alternative profit options with nodes in the loops. Similarly as for a blossom we have that eventually the dynamics becomes linear as stated in the following:

\begin{proposition} For a bicycle with $l$ and $m$ matched edges in loops and $n$ matched edges in the cross-bar, there exists a time $T_0\geq 0$ such that for every $t\geq T_0$, $(\vec{z}(t),\vec{x}(t), \vec{y}(t))$ evolves according to a linear system. Furthermore, $T_0 = O(\max(l,m)^2)$.
\end{proposition}

This observation follows from Theorem~\ref{lem:y1m} applied to each loop of the bicycle. This can be done because both $y_1(t) + y_m(t)$ and $z_1(t) + z_l(t)$ evolve autonomously as given by Lemma~\ref{lem:sumy} for $y_1(t)+y_m(t)$ and analogously for $z_1(t)+z_l(t)$.

We have shown that the dynamics for a bicycle is eventually according to a linear system, which is specified as follows:
\begin{equation}
\left(\begin{array}{c}\vec{z}(t+1)\\\vec{x}(t+1)\\\vec{y}(t+1)\end{array}\right) = \matA \left(\begin{array}{c}\vec{z}(t)\\\vec{x}(t)\\\vec{y}(t)\end{array}\right) + \vec{b}
\label{equ:bicyclelin}
\end{equation}
where 
$$
\matA=
\left(
\begin{array}{ccc}
\matT_l & \matQ' & \matnull\\
\matP' & \matT_n & \matP\\
\matnull & \matQ & \matT_m
\end{array} 
\right)
$$
with the given matrix blocks defined by 
$$
\left(
\begin{array}{cc}
\matT_l & \matQ'\\
\matP' & \matT_n\\
\end{array} 
\right)
\hbox{ and }
\left(
\begin{array}{ccc}
\matT_n & \matP\\
\matQ & \matT_m
\end{array} 
\right)
$$
are the matrices that correspond to two blossoms formed by loop~1 and cross-bar, and cross-bar and loop~2, respectively. 

The pair $(\matA,\vec{b})$ admits four possible values, corresponding to all possible combinations of two cases for each of the loops (Case~1 and Case~2 in Section~\ref{sec:blossom}): 
\begin{enumerate} 
\item Both $\matP'$ and $\matP$ as in Case~1\\$\vec{b} = (1/2,\underbrace{0,\ldots,0}_{l+n+m-2},1/2)^T$,
\item $\matP'$ as in Case~1, $\matP$ as in Case~2,\\ $\vec{b} = (1/2,\underbrace{0,\ldots,0}_{l-1},1/2,\underbrace{0,\ldots,0}_{m + n-2},1/2)^T$,
\item $\matP'$ as in Case~2, $\matP$ as in Case~1,\\ $\vec{b} = (1/2,\underbrace{0,\ldots,0}_{l + n - 2},-1/2,\underbrace{0,\ldots,0}_{m-1},1/2)^T$,
\item Both $\matP'$ and $\matP$ as in Case~2,\\ $\vec{b} = (1/2,\underbrace{0,\ldots,0}_{l-1},1/2,\underbrace{0,\ldots,0}_{n-2},-1/2,\underbrace{0,\ldots,0}_{m-1},1/2)^T$.
\end{enumerate}

In the following, we will only consider the case under item~1 as the same end results hold for other cases. The eigenvalues of the matrix $\matA$ are given by
\begin{eqnarray*}
&& \cos\left(\frac{\pi (2k-1)}{l+1}\right), \:k=1,\ldots, \lceil l/2 \rceil,\\ 
&& \cos\left(\frac{\pi (2k-1)}{m+1}\right), \:k=1,\ldots, \lceil m/2 \rceil,\\
&& \cos\left(\frac{2\pi k}{2n+l+m}\right), \:k=1,\ldots, n+\lfloor l/2\rfloor+\lfloor m/2\rfloor
\end{eqnarray*}
which we establish in Section \ref{sec:eig_bicycle}.

Remark that in any case all the eigenvalues are strictly smaller than $1$. On the other hand, if both $l$ and $m$ are even, then $-1$ is an eigenvalue with eigenvector $(1,-1,1,-1,\dots,1,-1)^T$, and otherwise, all the eigenvalues are strictly larger than $-1$. Therefore, if both $l$ and $m$ are even, then the asymptotic behavior of system (\ref{equ:bicyclelin}) is periodic, while otherwise, it is globally asymptotically stable.

As a byproduct, similarly to Theorem \ref{th_blossom}, we can establish that from an instance at which the system became linear, the convergence time scales as follows.

\begin{theorem} For a bicycle with $n$ matched edges in the stem and $m$ and $l$ matched edges in the loops, the convergence time $T$ satisfies the following. If $m$ or $l$ is even, then for every $0 < \alpha < 1$,
$$
T = \frac{2}{\alpha\pi^2}(2n+m+l)^2 \cdot [1 + o(1)]
$$
otherwise, if $m$ and $n$ are odd, then for every $0 < \alpha \leq 1$,
$$
T = \frac{2}{\alpha\pi^2} \max\left(m^2,l^2, \frac{1}{4}(2n + m+ l)^2\right) \cdot [1 + o(1)].
$$
\end{theorem}

Therefore, the convergence time is $O((l + n + m)^2)$, i.e. quadratic in the number of matched edges.

%This is indeed true in the edge balanced dynamics, as we proceed inductively in phase, where each phase we fix the outcome values for the nodes in the KT structure $C_l$ with slack $\sigma_l$ before moving to $C_{l+1}$ with slack $\sigma_{l+1}$,  $\sigma_{l}\leq\sigma_{l+1}$, for $l=0,\dots,k-1$. 

\section{Related Work}
\label{sec:related}

The concept of balanced outcomes was introduced by Nash in \cite{N50} for the case of two players with exogenous profit options. This concept follows from a set of axioms and different axioms were subsequently considered; e.g. see \cite{M88}. 

Kleinberg and Tardos~\cite{KT08} considered the concept of Nash bargaining solutions on graphs where profit options available to a player are not exogenously given but determined by her position in the graph. They established relations between stable and balanced outcomes and devised a polynomial time algorithm for computing balanced outcomes. Their work left open the question on existence and properties of local dynamics.

A local dynamics for Nash bargaining on graphs was recently considered by Azar et al~\cite{ABCDP09}. This paper assumed a fixed matching of nodes and considered a local, so called edge-balanced dynamics, for outcome vector $\vec{x}$. They established that fixed points of this dynamics are balanced outcomes and that the convergence to the fixed point occurs in an exponential number of rounds. A concurrent and independent work by Celis, Devanur and Peres~\cite{CDP10} considered the rate of convergence of edge-balanced dynamics. Their approach is different from ours in using a reduction of edge-balanced dynamics to a random-turn game for a class of graphs with uniform edge weights. For this class of graphs, their convergence time is quadratic in the maximum path length of an auxiliary graph derived from the input graph and given matching. This class of graphs does not accommodate cycles, but accommodates paths, blossoms and bicycles and for these cases the bound is quadratic in the number of matched edges. Another difference with our work is that for each of our elementary subgraphs we provide explicit characterization of dynamics and tight asymptotic estimate of convergence rates (exact constant factors). %A concurrent follow-on  \cite{CDP10} work proves a convergence time which is quadratic in the size of the maximum matching. The technique of the proof is based on an analogy with random-turn-games.

The assumption that matching is fixed was removed by Kanoria et al by introducing a natural dynamics studied in~\cite{KBBCM10} and \cite{KBBCM10b}. It was shown in \cite{KBBCM10} that provided that there exists a unique Nash bargaining solution and the graph satisfies the positive gap condition, the natural dynamics converges to this Nash bargaining solution in a polynomial time. Specifically, they showed that there exists a constant $C > 0$ such that the convergence time is upper bounded by $C[W/\sigma + \log(\sigma/\epsilon)] n^{6+\delta}$, where $W$ is an upper bound on the maximum edge weight, $\sigma > 0$ is the gap and $\epsilon,\delta > 0$. In \cite{DV10}, we showed, using techniques similar to the ones introduced in this paper, that the bound in the number of nodes can be improved to $O( n^{4+\delta})$. Using a different approach, in \cite{KBBCM10b}, the authors established that if the maximum weight matching is unique, then there exists $T = O(n^4/g^2)$ such that for every initial value the natural dynamics induces the maximum-weight matching, for every $t\geq T$; where $g$ is the difference between the weight of the maximum-weight matching and the weight of the second best matching (we refer to as matching weight gap).

Finally, another related work is that on maximum-weighted matchings on graph. There indeed  is a close connection between stable outcomes and maximum weight matchings reflected by the similarity of the distributed algorithms considered for solving both problems. In particular, Bayati et al~\cite{BSS08} considered an auction-like algorithm, which is similar in spirit to the natural dynamics for solving the balanced allocation problem, and showed that for complete bipartite graphs with a unique maximum-weight matching, the convergence time is $O(W n/g)$ where $W$ is the maximum edge weight, $g$ is the matching weight gap and $n$ is the number of nodes.

\section{Conclusion}
\label{sec:conc}

In this paper we showed that some known Nash bargaining dynamics on graphs can (eventually) be characterized by linear dynamical systems and this enabled us to derive tight characterizations of their convergence rates. Note that if the dynamics, as restricted to each of the KT elementary graphs that arise in the KT decomposition, were decoupled then the previous analysis will yield $O(n^3)$ convergence time, since there are at most $n$ such structures each taking $O(n^2)$ time to converge.

An interesting direction for future work is to investigate the extent by which the dynamics on the different KT substructures are coupled under assumptions such as the positive gap condition of the KT decomposition or the matching weight gap. Another interesting direction is to analyze the bargaining power of nodes based on their network position.

barganing arxiv

\section{Appendix}

\subsection{Eigenvalues for a Blossom}
\label{sec:eigbloss}

Remark that an eigenvalue $\lambda$ and eigenvector $\vec{v}$ of matrix $\matA$ satisfy $\matA \vec{v} = \lambda \vec{v}$, i.e.
\begin{eqnarray}
\frac{1}{2}v_2 &=& \lambda v_1\label{equ:z1}\\
\frac{1}{2}v_{i-1} + \frac{1}{2}v_{i+1} &=& \lambda v_i, \ 1 < i < n+m\label{equ:z2}\\
-\frac{1}{2}v_n + \frac{1}{2}v_{n+m-1} &=& \lambda v_{n+m}\label{equ:z3}
\end{eqnarray}
Suppose $\lambda = \cos(\phi)$ and $\vec{v} = (\sin(\phi), \sin(2\phi),\ldots,\sin((n+m)\phi))^T$. Then, by elementary trigonometric identities we note that (\ref{equ:z1}) and (\ref{equ:z2}) hold for every $\phi$. On the other hand, (\ref{equ:z3}) is equivalent to
$$
\sin((n+m-1)\phi) = \sin(n\phi) + 2 \cos(\phi)\sin((n+m)\phi)
$$ 
%Using the elementary trigonometric identities
%\begin{eqnarray*}
%&& 2 \cos(\phi)\sin((n+m)\phi)\\
%&=& \sin((n+m-1)\phi) + \sin((n+m+1)\phi)\\
%\end{eqnarray*}
%and
%\begin{eqnarray*}
%&& \sin(n\phi) + \sin((n+m+1)\phi)
%&=& 2 \sin\left(\frac{2n+m+1}{2}\phi\right)\cos\left(\frac{m+1}{2}\phi\right)\\
%\end{eqnarray*}
which by using elementary trigonometric identities is equivalent to
$$
\sin\left(\frac{2n+m+1}{2}\phi\right)\cos\left(\frac{m+1}{2}\phi\right) = 0.
$$
Therefore, $\phi$ is either
$$
\phi_1 = \frac{2k_1}{2n+m+1}\pi\ \hbox{ or }\ \phi_2 = \frac{2k_2-1}{m+1}\pi
$$
where $k_1$ and $k_2$ are arbitrary integers. Since cosine is a periodic function, it can be readily checked that $\cos(\phi_1)$ attains all possible values over $k = 1,2,\ldots, n + \lfloor m/2 \rfloor$ and similarly for $\cos(\phi_2)$ over $k = 1,2,\ldots, \lceil m/2\rceil$.

\subsection{Eigenvalues for a Bicycle}
\label{sec:eig_bicycle}
If $\lambda$ is an eigenvalue of matrix $\matA$ with eigenvector $\vec{v}$, then we have
%\iftr
\begin{eqnarray*}
\lambda v_1&=& \frac{1}{2}v_2- \frac{1}{2}v_{l+1}\\
\lambda v_i&=& \frac{1}{2}v_{i-1}+ \frac{1}{2}v_{i+1}, i=2,\ldots, n+m+l-1\\
\lambda v_{n+m+l}&=& \frac{1}{2}v_{n+m+l-1}-\frac{1}{2} v_{l+n}\:.
\end{eqnarray*}
%\else
%\begin{eqnarray*}
%\lambda v_1&=& \frac{1}{2}v_2- \frac{1}{2}v_{l+1}\\
%\lambda v_i&=& \frac{1}{2}v_{i-1}+ \frac{1}{2}v_{i+1}, i=2,\ldots, n+m+l-1\\
%\lambda v_{n+m+l}&=& \frac{1}{2}v_{n+m+l-1}-\frac{1}{2} v_{l+n}\:.
%\end{eqnarray*}
%%\begin{eqnarray*}
%%\lambda v_1&=&(v_2-v_{l+1})/2\\
%%\lambda v_i&=&(v_{i-1}+v_{i+1})/2,\quad i=2,\ldots, n+m+l-1\\
%%\lambda v_{n+m+l}&=&(v_{n+m+l-1}-v_{l+n})/2\:.
%%\end{eqnarray*}
%\fi

In the remainder, we separately consider two cases depending on whether either $l$ or $m$ is even, or otherwise.

\noindent
{\bf Case 1}: $l$ or $m$ is odd. 

Without loss of generality, suppose $l$ is odd. Let us introduce the following one-to-one linear transformation $\vec{z} = \matS \vec{v}$ where matrix $\matS$ is defined by $z_i=v_i+v_{l-i+1}$, for $i=1,\ldots,\lfloor l/2 \rfloor$, and $z_i=2v_{i}$ for $i=\lfloor l/2 \rfloor+1,\dots, n+l+m$. It is not difficult to verify that $\matS$ is non-singular and thus a matrix $\matB$ such that $\matA = \matS^{-1}\matB \matS$ is similar to $\matA$ and, therefore, has the same eigenvalues~\cite{HJ85}[Theorem~1.3.3]. 

Using the transformation $\vec{z} = \matS \vec{v}$ and $\matA \vec{v} = \lambda \vec{v}$, we have 
\begin{eqnarray*}
\lambda z_1&=& \frac{1}{2} z_2\\
\lambda z_{\lfloor l/2 \rfloor+1}&=& z_{\lfloor l/2 \rfloor}\\
\lambda z_{n+m+l}&=&\frac{1}{2}z_{n+m+l-1}-\frac{1}{2}z_{n+l}\:
\end{eqnarray*}
and for $i=2,\ldots, \lfloor l/2 \rfloor$ and $ i=\lfloor l/2 \rfloor+2,\ldots, n+l+m-1$, 
$$\lambda z_i=\frac{1}{2}z_{i-1}+\frac{1}{2}z_{i+1}\:.$$

Notice that $\lambda \vec{z} = \matS \matA \matS^{-1} \vec{z} = \matB \vec{z}$, and from the above identities
$$
\matB = \left(
\begin{array}{cc}
\matP & \matnull\\
\matQ & \matR
\end{array}
\right)
$$
where $\matP$ is a $\lceil l/2\rceil \times \lceil l/2\rceil$ tridiagonal matrix given by
$$
\matP=\left( \begin{array}{ccccc}
0 & 1/2 & 0& \cdots &0\\
1/2 & \ddots  & \ddots &\ddots &\vdots\\
0 & \ddots &\ddots &\ddots& 0\\
\vdots & \ddots & 1/2 & \ddots & 1/2\\
0&\cdots& 0& 1 & 0 \end{array} \right)\:,
$$
$\matQ$ is a $(n+m+\lfloor l/2\rfloor)\times \lceil l/2\rceil$ matrix with all elements equal to zero but the element in the first row and last column equal to $1/2$, and $\matR$ is a $(n+m+\lfloor l/2\rfloor) \times (n+m+\lfloor l/2\rfloor)$ matrix that corresponds to a blossom with $n+\lceil l/2\rceil$ matched stem edges and $m$ matched loop edges and is of the form (\ref{equ:Ablossom}) under Case~1 in Section~\ref{sec:blossom}.

Using the properties of determinants of block matrices, we observe that eigenvalues of $\matB$ consist of eigenvalues of matrices $\matP$ and $\matR$. Therefore, the eigenvalues of matrix $\matB$, and by similarity of matrix $\matA$, are
\begin{eqnarray}
& \cos\left(\frac{\pi (2k-1)}{l+1}\right), \:k=1,\ldots, \lceil l/2 \rceil,\label{equ:l1}\\ 
& \cos\left(\frac{\pi (2k-1)}{m+1}\right), \:k=1,\ldots, \lceil m/2 \rceil,\label{equ:l2}\\  
& \cos\left(\frac{2\pi k}{2n+l+m}\right), \:k=1,\ldots, n+\lfloor l/2\rfloor+\lfloor m/2\rfloor\label{equ:l3}
\end{eqnarray}
where (\ref{equ:l1}) are eigenvalues of matrix $\matP$, which is easily derived and thus omitted, and (\ref{equ:l2}) and (\ref{equ:l3}) are eigenvalues of $\matR$ which we have already showed in Section~\ref{sec:blossom}.

It is not difficult to see that the above eigenvalues hold whenever either $l$ or $m$ is odd.

\noindent
{\bf Case 2}: both $l$ and $m$ are even.

We use a similar but different one-to-one transformation as under Case~1: $z_i=v_i+v_{l-i+1}$, for $i=1,\ldots, l/2$, $z_i=v_{i}+v_{i+1}$ for $i=l/2+1,\ldots, n+l+m/2$, and $z_{i+n+l}=v_{n+l+i}+v_{m+n+l-i+1}$ for $i=0,\ldots, m/2$. We have that
\begin{eqnarray*}
\lambda z_1 &=& \frac{1}{2}z_2\\
\lambda z_{l/2} &=& \frac{1}{2}z_{l/2-1}+\frac{1}{2}z_{l/2}\\
\lambda z_{n+l+m/2} &=& \frac{1}{2}z_{n+l+m/2}+\frac{1}{2}z_{n+l+m/2+1}\\
\lambda z_{n+m+l} &=& \frac{1}{2}z_{n+m+l-1}\:
\end{eqnarray*}
and for $i= l/2+1,\ldots, n+l+m/2-1$ and $i=n+l+m/2+1,\ldots, n+m+l-1$, 
$$
\lambda z_i=\frac{1}{2}z_{i-1}+\frac{1}{2}z_{i+1}\:.$$

Similarly as for Case~1, using the properties of determinants of block matrices, we have that the eigenvalues of $\matA$ are
\begin{eqnarray*}
&& \cos\left(\frac{\pi (2k-1)}{l+1}\right), \:k=1,\ldots, l/2,\\ 
&& \cos\left(\frac{\pi (2k-1)}{m+1}\right), \:k=1,\ldots, m/2,\\  
&& \cos\left(\frac{\pi k}{n+l/2+m/2}\right), \:k=1,\ldots, n+ l/2+ m/2-1,\\
&& \hbox{and } -1
\end{eqnarray*} 
where $(1,-1,1,-1,\ldots,1,-1)^T$ is the eigenvector of eigenvalue $-1$.

\end{document}